\documentclass[a4paper,11pt]{article}

\usepackage[utf8]{inputenc} 
\usepackage[T1]{fontenc}
\usepackage{url}
\usepackage{ifthen}
\usepackage{cite}
\usepackage[cmex10]{amsmath}

\usepackage{hyperref}

\usepackage{amssymb, bbm}

\usepackage{amsthm}

\theoremstyle{plain}
\newtheorem{theorem}{Theorem}
\newtheorem{lemma}[theorem]{Lemma}
\newtheorem{construction}[theorem]{Construction}

\theoremstyle{definition}
\newtheorem{definition}[theorem]{Definition}

\theoremstyle{remark}
\newtheorem{remark}[theorem]{Remark}
\newtheorem{example}[theorem]{Example}

\newcommand{\bF}{\mathbb{F}}
\newcommand{\bfa}{\boldsymbol{a}}

\newcommand{\bfx}{\boldsymbol{x}} 
\newcommand{\bfy}{\boldsymbol{y}} 

\newcommand{\cS}{\mathcal{S}}

\DeclareMathOperator{\vts}{VTS}
\DeclareMathOperator{\wt}{wt}

\usepackage{authblk}

\author[1]{Michael Schaller}
\author[2]{Beatrice Toesca}
\author[3]{Van Khu Vu}

\affil[1,2]{University of Zurich}
\affil[3]{National University of Singapore}

{
    \makeatletter
    \renewcommand\AB@affilsepx{: \protect\Affilfont}
    \makeatother

    \affil[ ]{Email}

    \makeatletter
    \renewcommand\AB@affilsepx{, \protect\Affilfont}
    \makeatother

    \affil[1]{michael.schaller@math.uzh.ch}
    \affil[2]{beatrice.toesca@math.uzh.ch}
    \affil[3]{isevvk@nus.edu.sg}
}

\date{}

\title{A New Construction of Non-Binary\\ Deletion Correcting Codes and their Decoding}

\begin{document}

\maketitle

\begin{abstract}
    Non-binary codes correcting multiple deletions have recently attracted a lot of attention.
    In this work, we focus on multiplicity-free codes, a family of non-binary codes where all symbols are distinct.
    Our main contribution is a new explicit construction of such codes, based on set and permutation codes.
    We show that our multiplicity-free codes can correct multiple deletions and provide a decoding algorithm.
    We also show that, for a certain regime of parameters, our constructed codes have size larger than all the previously known non-binary codes correcting multiple deletions.
\end{abstract}

\section{Introduction}\label{sec:intro}
Codes correcting synchronization errors have a long history from the seminal paper by Levenshtein in 1966 \cite{lcvenshtcin1966binary}, where a binary code correcting a single deletion was presented.
Levenshtein already established an interesting connection between codes over multiplicity-free sets and ordered partial Steiner systems \cite{levenshtein1992perfect}.
In 1984, Tenengolts presented the first construction of a non-binary code correcting a single deletion \cite{tenengolts1984nonbinary}. 

From a practical point of view, non-binary deletion correcting codes have applications in DNA-based storage systems and racetrack memories.
In DNA-based storage systems, quaternary codes correcting multiple deletions are required \cite{cai2021correcting}.
In racetrack memories, the non-binary deletion correcting codes have a large alphabet size depending on the number of read-heads in the racetrack memory. \cite{chee2018codes,chee2021locally}. 
For example, in order to correct over-shift errors in racetrack memories, a non-binary code correcting multiple blocks of deletions was proposed \cite{chee2021locally}. 

Recently, there were numerous breakthrough results both for binary \cite{gabrys2018codes,sima2019two,sima2020optimalbinary} and 
non-binary \cite{haeupler2021synchronization, liu2022bounds,sima2020optimalqary, song2023non, con2023reed, heckel2019characterization, chee2018coding} deletion correcting codes.
Other papers consider the case that there is a given number of deletions \cite{sima2020optimalqary, liu2022bounds} or that there is a large number of deletions \cite{haeupler2021synchronization, liu20212}.
\vspace*{1mm}

To explain better some of these results and our contribution to the paper, we introduce some notation.
Let $q,n$ be two integers, $\Sigma_q=\{0,1,\ldots,q-1\}$ be the alphabet of size $q$ and $[n]$ be the set $\{1,2,\ldots,n\}$.
Let $\Sigma_q^n$ be the set of all $q$-ary sequences of length $n.$
For $n\leq q$, a sequence $\bfx=(x_1,\ldots,x_n) \in \Sigma_q^n$ is called a \textit{multiplicity-free sequence} of length $n$ in the alphabet $\Sigma_q$ if $x_i \neq x_j$ for any given pair $i \neq j$.
We denote by $M_q^n$ the set of all multiplicity-free sequences of length $n$ in $\Sigma_q$.

Let $C(q,n,t)$ denote a $q$-ary code of length $n$ that can correct $t$ deletions.
We call  $\log_2 |\Sigma_q^n| - \log_2 |C(q, n, t)|$ the \textit{bit redundancy} or \textit{redundancy} of the code $C(q,n,t)$.
\vspace*{2mm}

The work in \cite{sima2020optimalbinary,chee2021two, con2023reed, liu2022bounds} shows that there are a few families of non-binary codes correcting multiple deletions approaching the Singleton bound for some regime of parameters.
In \cite{chee2021two, HaeuplerS21}, the authors used a sequence of length $n$ of small alphabet size to locate all deletion errors and an erasure correcting code to correct them.
In total, they need at least $t\log q + \Theta(n)$ bits of redundancy to construct a non-binary code that can correct $t$ errors. In \cite{liu2022bounds}, when the alphabet size $q$ is exponential in the code length $n$, the authors presented a construction of codes approaching the Singleton bound. 

To the best of our knowledge, in the regime in which the alphabet size $q$ is a power of $n$, there is no known construction arbitrarily close to the Singleton bound for fixed $t$.

For any $q>n$, there is a construction of $q$-ary codes correcting $t$ deletions with at most $30t \log q$ bits of redundancy \cite{sima2020optimalqary}.
In the recent preprint \cite{hagiwara2024pointer}, Hagiwara and Vu provided a construction of codes over multiplicity-free sets with an efficient encoding algorithm.
They obtained a family of $q$-ary codes of length $n$ correcting $t$ deletions with at most $5t \log q$ bits of redundancy.

Our objective is to construct a new class of $q$-ary codes of length $n$ correcting $t$ deletions for any $q>n$.
The main result provides a new construction of codes over multiplicity-free sets based on results in permutation codes and set codes.
In the case of large $n$ and $q> n^{2 + \varepsilon}$ for arbitrary positive $\varepsilon$, our codes have redundancy at most $t \log(q) + (3t - 1) \log(n) + \delta t + O\left(\frac{1}{n^{\varepsilon}} \right)$ for arbitrary $\delta > 0$.

The construction we derive is based on a decomposition of a multiplicity-free sequence into a permutation and a set of symbols.
We observe that if there is a set code that can correct multiple deletions and a permutation code that can correct multiple unstable deletions, we can obtain a multiplicity-free code that can correct multiple deletions.
The idea also works for stable deletions in permutation codes.
\vspace*{2mm}

The paper is structured as follows:
In the next section, we review the results on set codes and binary constant weight codes. 
In Section \ref{sec.perm_codes}, we provide definitions and recent results on permutation codes that can correct multiple stable/unstable deletions.
In Section \ref{sec:mult-free_codes}, we provide the main results of the paper by combining permutation codes and set codes to obtain a new family of multiplicity-free codes that can correct multiple deletions.
In Section \ref{subsec:analysis}, we analyze the size of our codes and compare it with previous results in the literature.

\section{Set Codes and Binary Constant-Weight Codes}

\subsection{Set Codes} \label{sec:set_codes}
A set $A \subset \Sigma_q$ is called an \textit{$n$-subset} of the alphabet $\Sigma_q$ if it has $n$ elements, that is $|A|=n$. We denote with $\binom{\Sigma_q}{n}$ the set of all $n$-subsets of $\Sigma_q$.
In the following, we introduce the notion of \textit{set codes}, which will play a crucial role in our construction of non-binary deletion correcting codes.

\begin{definition}
    Given an alphabet $\Sigma_q$, a \textit{set code} of length $n$ is a family $C_S(q,n) \subseteq \binom{\Sigma_q}{n}$.
\end{definition}
Note that a set code can also be interpreted as an $n$-uniform hypergraph, where each edge corresponds to a codeword of the set code.
\begin{definition}
    A set code $C_S(q,n)$ is said to \textit{correct $t$ deletions} if there are no two sets in $C_S(q,n)$ that result in the same set after at most $t$ deletions.
\end{definition}
This tells us that after $t$ deletions, the codeword can still be uniquely recovered.
Therefore, we are dealing with a partial Steiner system with parameters $S(n-t, n, q)$.
The problem of maximizing the size of a set code correcting deletions, then corresponds to the problem of maximizing the number of blocks in the corresponding partial Steiner system.
There are several results on the size of partial Steiner systems, see in particular \cite{Rodl85, Kuzjurin95}.
However, these results give only probabilistic constructions and no decoding algorithms. Moreover, they mainly provide results in the case when the deletion correction capability $t$ is close to the length $n$ of the set code.

In the following, we explain the correspondence of set codes with binary constant-weight codes.
In Section \ref{sec:bcw_codes}, we will then exploit this correspondence to explicitly create a set code correcting $t$ deletions.

\subsection{Correspondence between Set Codes and Binary Constant-Weight Codes} \label{sec:set_bcw_codes}
A set of vectors in $\bF_2^q$ is called a \textit{binary constant-weight code} of length $q$ and weight $n$ if all vectors share the same Hamming weight $n$.
They have been studied for a long time, for example in \cite{BrouwerSSS90, GrahamS80}.

We can associate to a set code $C_S(q,n)\subseteq \binom{\Sigma_q}{n}$ a binary constant-weight code as follows.
Fix an ordering of the alphabet $\Sigma_q$ and take as length of the constant-weight code the alphabet size $q$.
For every $n$-subset in the set code, we associate to it a vector of length $q$ that contains a $1$ in every position corresponding to an alphabet symbol which is part of the $n$-subset, and a $0$ in every other position.
Notice that each vector constructed this way will contain exactly $n$ ones, so the resulting binary code has constant weight $n$.

Note that deleting an element in a codeword of the set code corresponds to substituting a $1$ with a $0$ in the corresponding position of the codeword of the binary constant-weight code. Since we are only concerned with set codes with respect to deletions and not insertions, we will only be interested in having a binary constant-weight code that can correct asymmetric errors from 1 to 0 and not the opposite.

The problem of maximizing the size of a set code $C_S(q,n)$ correcting $t$ deletions, then corresponds to the problem of maximizing the size of a binary constant-weight code of length $q$ and weight $n$ correcting $t$ asymmetric errors.

Note that the map that gives a binary constant-weight code from a set code and its inverse can both be computed in time $O(q)$, which is efficient.

\subsection{VT-Syndrome Binary Constant-Weight Codes} \label{sec:bcw_codes}
In this subsection, we present a generalization of the Varshamov-Tenengolts code proposed in \cite{varshamov1965code}, which could only correct a single error.
We do so by adapting the construction from \cite{dolecek2010towards} following \cite{sabary2024error}.
This is very similar to the constructions of the Graham-Sloane bounds in \cite{BrouwerSSS90} and \cite{GrahamS80}, and is related to sets with distinct subset sums.
We focus on the construction from \cite{dolecek2010towards}, since it also allows for efficient decoding in the case of asymmetric errors.

This binary constant-weight code correcting multiple asymmetric errors can then be turned into a set code correcting deletions using the correspondence in Section \ref{sec:set_bcw_codes}.

We define a generalization of the Varshamov-Tenengolts syndrome.
\begin{definition}
    Let $\bfx \in \{0, 1\}^q$ and $p$ be a prime with $p > q$.
    Then for an integer $t \geq 1$ we define the \textit{VT-syndrome vector} of $\bfx$ as
    \[
        \vts_t(\bfx) = \left(\sum_{i=1}^q i x_i \bmod p,\sum_{i=1}^q i^2 x_i \bmod p,\ldots,\sum_{i=1}^q i^t x_i \bmod p \right). 
    \]
\end{definition}
This enables us to define the following binary code.
\begin{definition}
    Let $p$ be a prime with $p > q$, $t \geq 1$ an integer, and $\bfa \in \bF_p^t$.
    We define a code using the VT-syndrome vector:
    \[
        C(q, n, p, t, \bfa) = \left\{\bfx \in \{0, 1\}^q : \wt(\bfx) = n, \vts_t(\bfx) = \bfa \bmod p\right\}.
    \]
\end{definition}

Notice that the code just defined is a binary constant-weight code of length $q$ and weight $n$. As shown in \cite{dolecek2010towards}, such a code can correct $t$ asymmetric flips from $1$ to $0$. The decoding process from \cite{dolecek2010towards} has linear complexity in $q$ and $t$ up to polylogarithmic factors.

We now want to give a lower bound on the code size achieved with this construction.
For different values of $\bfa$, the codes $C(q, n, p, t, \bfa)$ partition the set of all binary vectors of length $q$ and weight $n$.
Since there are in total $\binom{q}{n}$ binary vectors of length $q$ and weight $n$, and we are partitioning them into $p^t$ classes, we get that there exists at least one class with size at least
$$\frac{\binom{q}{n}}{p^{t}}.$$
From Bertrand's postulate, we can choose the prime $p$ such that $q < p \leq 2q$.
Hence, for any given $q,n,t$, there exists a code of size at least
\[
    \frac{\binom{q}{n}}{(2q)^{t}}.
\]

\section{Permutation Codes Correcting Deletions}\label{sec.perm_codes}
Permutation codes were first introduced in \cite{slepian1965permutation} for coding over channels with Gaussian noise.
They more recently attracted a lot of attention in the setting of deletion errors due to their applications in flash memories.
More details about permutation codes can be found in \cite{levenshtein1992perfect, farnoud2013error, chee2014breakpoint, chee2019burst}.
In this section, we provide some known results on permutation codes and their behavior under deletions. In particular, we analyze two different kinds of deletions: stable and unstable deletions. 

Recall that a \textit{permutation} $\sigma$ is a bijection $\sigma : [n]\to [n]$, where $[n]$ denotes the set $\{1,\dots,n\}$. We write permutations as sequences $\sigma=(\sigma_1,\ldots,\sigma_n)$ of length $n$, where the meaning is that for every $i\in[n]$, $\sigma_i:=\sigma(i)$.
For a given length $n$, the \textit{symmetric group} $\cS_n$ is the set of all permutations over $[n]$.

\begin{definition}
    A \textit{permutation code} of length $n$ is a subset of the symmetric group $\cS_n$.
\end{definition} 

To construct permutation codes with large deletion correcting capability and large size, we first need to investigate how permutations behave with respect to deletions.

\subsection{Stable Deletions}\label{sec:stable_del}
Let $\sigma=(\sigma_1,\ldots,\sigma_n) \in \cS_n$ be a permutation of length $n$. In the case of stable deletions, once a symbol is deleted, all other symbols remain the same. If the deletion happens in the $i$-th position, the new sequence is $(\sigma_1,\dots,\sigma_{i-1},\sigma_{i+1},\dots,\sigma_n)\in[n]^{n-1}$. Note that the new sequence no longer represents a permutation.
\begin{example}
    Let $\sigma = (2, 3, 1, 4, 5) \in \cS_5$ and suppose there is a stable deletion in the second position. The new sequence obtained is $\sigma'=(2, 1, 4, 5)$.
\end{example}
\begin{remark}\label{rk:SD}
    In the case of multiple simultaneous deletions, one can give an equivalent definition of the model as follows.
    Let $\sigma=(\sigma_1,\ldots,\sigma_n)\in\cS_n$ and let $I \subset [n]$ be the set of positions where a deletion occurs.
    For every given integer $k \in [n]$ define $k(I)=k-|\{i \in I: i <k\}|$.   
    If the permutation $\sigma$ suffers $t$ stable deletions at the positions in set $I$, the resulting sequence is $\sigma'=(\sigma'_1,\ldots,\sigma'_{n-t})$, where for all $k \in [n] \setminus I$ and $i=k(I)$ we define $\sigma'_{i}=\sigma_k$.
\end{remark}
A permutation code of length $n$ that can correct $t$ stable deletions is called a $t$-SD correcting permutation code and is denoted by $C_{SD}(n,t)$. We are interested in maximizing the size of these codes with respect to the parameters $n,t$.

Recently, Wang et al. \cite{WNCV24} used techniques from extremal graph theory to prove the existence of a $t$-SD correcting permutation code of length $n$ with size $\Omega_t\left(\frac{n! \log n}{n^{2t}}\right)$.
However, the proof is not constructive, and no way of designing permutation codes with such size is currently known.

In \cite{WNCV24}, the authors presented a $t$-SD correcting permutation code of length $n$ with size $|C_{SD}(n,t)| \geq \frac{n!}{(2n)^{3t-1}}$.
This is, to the best of our knowledge, the largest known $t$-SD correcting permutation code with an explicit construction and decoding algorithm.
However, the algorithm relies on a decoding algorithm for permutation codes in the Hamming metric.
Unfortunately, for the best known permutation codes in the Hamming metric there does not yet exist an efficient decoding algorithm.

We will then use the code from \cite{WNCV24} in Section \ref{sec:mult-free_codes}, in combination with a set code, to obtain a multiplicity-free code correcting deletions with a decoding algorithm.

\subsection{Unstable Deletions}
Let $\sigma=(\sigma_1,\ldots,\sigma_n)$ be a permutation in $\cS_n$.
In case of an unstable deletion, after one symbol is deleted, the values of the others also change.
What is preserved is the relative order of the symbols, but their value is adjusted to have all the elements of $[n-1]$ appear exactly once.
If $i\in[n]$ is the position where the deletion occurs, the new sequence is $\tilde\sigma=(\tilde\sigma_1, \dots, \tilde\sigma_{i-1}, \tilde\sigma_{i+1}, \dots, \tilde\sigma_n) \in[n-1]^{n-1}$, where for every index $j$ we have $\tilde\sigma_j=\sigma_j-\mathbbm{1}_{\sigma_j>\sigma_i}$. Notice that the new sequence $\tilde\sigma\in\cS_{n-1}$ is again a permutation in a different symmetric group.
Similarly, if multiple unstable deletions occur and $I \subset [n]$ is the set of positions of deleted symbols, the new sequence will be a permutation $\tilde\sigma \in \cS_{n-|I|}$.
\begin{example}
    Let $\sigma = (2, 3, 1, 4, 5) \in \cS_5$ and suppose there is an unstable deletion in position $i=2$. The new sequence obtained is $\tilde\sigma=(2, 1, 3, 4)\in\cS_4$. Starting again from the original permutation $\sigma$, suppose now that two unstable deletions occur in positions $I=\{2,4\}$. The resulting sequence is $\tilde\sigma=(2, 1, 3) \in \cS_3$.
\end{example}
\begin{remark}
    Let $\sigma=(\sigma_1,\ldots,\sigma_n)\in\cS_n$ and $I \subset [n]$ be the set of deleted positions.
    For $k \in [n]$, let $k(I)$ be as in Remark \ref{rk:SD}.
    Define $\sigma(I)=\{\sigma_i: i \in I \}$.
    If the permutation $\sigma$ suffers $t$ unstable deletions at all positions in $I$,
    we obtain the permutation $\tilde\sigma = (\tilde\sigma_1 ,\dots, \tilde\sigma_{n-t}) \in\cS_{n-t}$, where for all $k\in[n] \setminus I$ and $i = k( I )$ we define $\tilde\sigma_i = \sigma_k (\sigma ( I ))$.
    In the language from Section \ref{sec:mult-free_codes}, this corresponds to the induced permutation of the vector resulting from the stable deletion.
\end{remark}

A permutation code of length $n$ that can correct $t$ unstable deletions is called a $t$-UD correcting permutation code and is denoted by $C_{UD}(n,t)$.

This setting of unstable deletions arises every time that only the relative order of the symbols is relevant, but not their value, and has the advantage of leading to a new permutation.
Permutation codes correcting unstable deletions have been investigated in \cite{gabrys2015codes,chee2015permutation}. However, less is known about them in comparison with the stable deletion setting, as only codes correcting a single unstable deletion or a burst of unstable deletions are known.

\section{Multiplicity-free Codes Correcting Multiple Deletions}\label{sec:mult-free_codes}
In this section, we combine the results of the previous sections to construct a multiplicity-free $q$-ary code of length $n$ correcting $t$ deletions, where $t$ is given and $q>n$.

\subsection{Construction}
To design our code, we first decompose a multiplicity-free sequence into a set and a permutation. We will show that if, after multiple deletion errors, we recover the original set and permutation, then we are also able to recover the original multiplicity-free sequence.
\begin{definition}
    The \textit{induced set} of a multiplicity-free sequence $\bfx \in M_q^n$ is the set $$A(\bfx)=\{x_i : i\in[n]\}.$$
\end{definition}

For every $\bfx \in M_q^n$, we now want to define a permutation $\sigma\in\cS_n$ such that, if we rearrange the elements of $\bfx$ starting from the increasing order and following the order given by $\sigma$, the resulting multiplicity-free sequence is $\bfx$.
Formally, we do the following:
\begin{definition}
    Let $\bfx = (x_1, \ldots, x_n)$ and let $i_1,\dots,i_n\in[n]$ be the reordering of the indices such that $x_{i_1} < x_{i_2} < \ldots < x_{i_n}$. For every $j\in[n]$, we define $\sigma(i_j)=j$.
    We say that $\sigma$ is the \textit{induced permutation} of $\bfx$ and we write $P(\bfx)=\sigma$.
\end{definition}
The statement $\sigma(i_j)=j$ asserts that the $j$-th smallest element in the set $A(\bfx)$ appears in the sequence $\bfx$ in position $i_j$.

\begin{example}
    If $\bfx= (8,0,6,5,2)$, then we have $A(\bfx)=\{0,2,5,6,8\}$ and $P(\bfx) = (5,1,4,3,2)$.
\end{example}

\begin{lemma}
    For each $\bfx \in M_q^n$, let $A(\bfx) \in \binom{\Sigma_q}{n}$ and $P(\bfx) \in \cS_n$ be the induced set and the induced permutation of $\bfx$. The following map is bijective
    \begin{align*}
    \Phi:   M_q^n & \to \binom{\Sigma_q}{n} \times \cS_n\\
      \bfx    & \mapsto A(\bfx)    \times P(\bfx).
    \end{align*}
\end{lemma}

\begin{proof}
    We show that the inverse map exists.
    We define $\Psi: \binom{\Sigma_q}{n} \times \cS_n \to M_q^n$ as follows.
    Let $A=\{a_1,\dots,a_n\}\in \binom{\Sigma_q}{n}$ be an $n$-set such that $a_1<\dots<a_n$ and let $\sigma=(\sigma_1,\dots,\sigma_n)\in \cS_n$ be a permutation. Define
    \[
        \Psi\big(\{a_1,\dots,a_n\},(\sigma_1, \dots, \sigma_n)\big) := (a_{\sigma_1}, \dots, a_{\sigma_n}).
    \]
    
    It is clear that
    \begin{align*}
        (\Phi \circ \Psi)\big(\{a_1,\dots,a_n\},&(\sigma_1, \dots, \sigma_n)\big) = \Phi\big((a_{\sigma(1)}, \dots, a_{\sigma(n)})\big)\\
        &= \big(\{a_1,\dots,a_n\},(\tau_1, \dots, \tau_n)\big)
    \end{align*} for the same $n$-set $A\in \binom{\Sigma_q}{n}$ and for some permutation $\tau\in \cS_n$.
    We now need to prove that the two permutations $\tau$ and $\sigma$ are the same.
    Note that, by the definition of the induced permutation map, $\tau(j) = \sigma(j)$ asserts that the $\sigma(j)$-th smallest element should appear in the $j$-th position of $(a_{\sigma(1)}, \dots, a_{\sigma(n)})$.
    Now, this is indeed true since at position $j$ in $(a_{\sigma(1)}, \dots, a_{\sigma(n)})$ there is the $\sigma(j)$-th smallest element.
    Hence $\tau = \sigma$ and this proves that $\Phi \circ \Psi = id_{\binom{\Sigma_q}{n}\times \cS_n}$.
    
    Next, we consider $$(\Psi \circ \Phi)\big((x_1, \ldots, x_n)\big) = \Psi \big(\{x_{i_1},\dots,x_{i_n}\}, (\sigma_1,\dots,\sigma_n) \big),$$
    where $x_{i_1}<\dots<x_{i_n}$ and $\sigma=P(\bfx)$.
    Then  $$\Psi\big(\{x_{i_1},\dots,x_{i_n}\}, (\sigma_1,\dots,\sigma_n) \big) = (x_{\sigma(i_1)},\dots,x_{\sigma(i_n)}).$$
    By definition of induced permutation, $\sigma(i_j)=j$ for every $j$. Hence, $(\Psi \circ \Phi)\big((x_1, \ldots, x_n)\big)=(x_1, \ldots, x_n)$ for every starting sequence $\bfx$ and $\Psi \circ \Phi = id_{M_q^n}$.
\end{proof}

The fact that $\Phi$ is bijective shows that, given a set $B \in \binom{\Sigma_q}{n}$ and a permutation $\pi \in \cS_n$, there is a unique sequence $\bfx \in M_q^n$ such that $A(\bfx)=B$ and $P(\bfx)=\pi$.

We can now present our first construction of a multiplicity-free code.
\begin{construction} \label{con1-UD}
    Let $C_S(q,n,t) \subseteq \binom{\Sigma_q}{n}$ be a set code of length $n$ correcting $t$ deletion errors and let $C_{UD}(n,t) \subseteq \cS_n$ be a permutation code of length $n$ correcting $t$ unstable deletions. Then $C_1(q,n,t) = \Phi^{-1} \big( C_S(q,n,t) \times C_{UD}(n,t) \big)  \subseteq M_q^n$ is a multiplicity-free code.
\end{construction}

\begin{theorem}
    The code $C_1(q,n,t)$ from Construction \ref{con1-UD} is a $q$-ary multiplicity-free code of length $n$ correcting $t$ deletions of size $|C_1(q,n,t)| = |C_S(q,n,t)|\cdot|C_{UD}(n,t)|$.
\end{theorem}
\begin{proof}
    The size of the code follows immediately from the fact that $\Phi$ is a bijection.
    All we have to prove is that the code corrects $t$ deletions.
    Given a sequence $\bfx \in C_1(q,n,t)$, let $\bfy \in M_q^{n-t}$ be a sequence obtained from $\bfx$ after $t$ deletions.
    The induced set $A(\bfy)$ is obtained from $A(\bfx)$ by deleting the $t$ elements that no longer appear in $\bfy$.
    Moreover, the induced permutation $P(\bfy) \in \cS_{n-t}$ can be obtained from $P(\bfx)$ after $t$ unstable deletions, as when computing the induced permutation only the relative ordering of the elements matters.
    After decoding separately the original induced set $A(\bfx)$ in the set code and the induced permutation $P(\bfx)$ in the permutation code, we can then apply the map $\Phi^{-1}$ and decode the original codeword $\bfx$.
    Hence, if we have a set code correcting $t$ deletions and a $t$-UD correcting permutation code, we can construct a multiplicity-free code correcting $t$ deletions.
\end{proof}

In Section \ref{sec:bcw_codes}, we presented an explicit construction of a binary constant-weight code of size at least $\frac{\binom{q}{n}}{(2 q)^t}$, which can be turned into a set code of the same size using the correspondence in Section \ref{sec:set_bcw_codes}.
However, there is a lack of knowledge on permutation codes correcting multiple unstable deletions, so Construction \ref{con1-UD} cannot be used for $t\geq 2$.

Fortunately, there are many good permutation codes correcting multiple stable deletions, as discussed in Section \ref{sec:stable_del}.
It is therefore desirable to design a $q$-ary multiplicity-free code based on a permutation code correcting stable deletions.
We do so in the following construction.

\begin{construction}\label{con2-SD}
    Let $C_S(q,n,t) \subseteq \binom{\Sigma_q}{n}$ be a set code of length $n$ correcting $t$ deletion errors and let $C_{SD}(n,t) \subseteq \cS_n$ be a permutation code of length $n$ correcting $t$ stable deletions. Then $C_2(q,n,t) = \Phi^{-1} \big( C_S(q,n,t) \times C_{SD}(n,t) \big)  \subseteq M_q^n$ is a multiplicity-free code.
\end{construction}
\begin{theorem} \label{thm.mult_free_stable_del}
    The code $C_2(q,n,t)$ from Construction \ref{con2-SD} is a $q$-ary multiplicity-free code of length $n$ correcting $t$ deletions of size $|C_2(q,n,t)| = |C_S(q,n,t)|\cdot|C_{SD}(n,t)|$.
\end{theorem}
\begin{proof}
    Recall that $\Phi$ is a bijection and thus $|C_2(q,n,t)| = |C_{S}(q,n,t)|\cdot|C_{SD}(n,t)|$.
    Now we prove that it can correct up to $t$ deletions.
    To simplify notation, we only consider the case of exactly $t$ deletions, but the proof stays the same for less than $t$ deletions.
    Let $\bfx \in C_2(q, n, t)$ and let $\bfy \in M_q^{n-t}$ be a sequence obtained from $\bfx$ after $t$ deletions.
    Again, since $A(\bfx)\in C_S(q,n,t)$ and the induced set $A(\bfy)$ is obtained from $A(\bfx)$ after $t$ deletions, we can recover the set $A(\bfx)$.
    
    Having recovered the induced set, we can do the following.
    We order the elements in $A(\bfx)$ as follows:
    \[
        x_{i_1} < x_{i_2} < \ldots < x_{i_n}.
    \]
    Hence, by definition of the induced permutation $P(\bfx)=\sigma$, for every index we have $\sigma(i_l) = l$.
    Then we define a vector $\tau$ such that, for every $j\in[n-t]$, if $\bfy_j = x_{i_l}$ we set $\tau_j = l$. Note that this is well-defined as all elements of $\bfy$ appear in $\bfx$ exactly once.
    Furthermore, observe that $\tau$ is now the vector obtained from the permutation $\sigma$ after the stable deletion of the $t$ entries in the same positions where there were deletions in $\bfx$.
    Since $\sigma=P(\bfx) \in C_{SD}(n,t)$, from the vector $\tau$ we can recover the original induced permutation $P(\bfx)$.
    From $A(\bfx)$ and $P(\bfx)$ we now recover $\bfx$ using $\Phi^{-1}$.
\end{proof}

The proof of Theorem \ref{thm.mult_free_stable_del} also gives an algorithm to decode, assuming we have decoding algorithms for the set code and the permutation code.
The main idea is to first use the decoder of the code $C_S(q,n,t)$ to recover the set $A(\bfx)$.
Once we have recovered $A(\bfx)$, we can use it to obtain the sequence $\tau$, which corresponds to $t$ stable deletions from $\sigma=P(\bfx)$.
Using the decoder of the code $C_{SD}(n,t)$, we can now recover the permutation $P(\bfx)$. Finally, from $A(\bfx)$ and $P(\bfx)$, we recover the original codeword $\bfx$.

\begin{example}
    Consider the $2$-SD correcting permutation code $$C_{SD}(5, 2) = \{(1,2,3,4,5), (4,5,2,3,1) \}$$ and the set code $$C_S(8, 5, 2) = \{\{0,1,2,3,4\}, \{3,4,5,6,7\} \}.$$
    Then, using Construction \ref{con2-SD}, we get the multiplicity-free code
    \begin{align*}
        C_2(q, n, t) = \{&(0,1,2,3,4), (3,4,1,2,0),\\
        &(3,4,5,6,7), (6,7,4,5,3) \}.
    \end{align*}
    If there are two deletions in positions $\{2, 4\}$ in the codeword $\bfx = (6,7,4,5,3)$, we get $\bfy = (6,4,3)$.
    The induced set is $A(\bfy) = \{3, 4, 6\}$, from which we recover the original set $A(\bfx) = \{3,4,5,6,7\}$.
    Having recovered the set, we can get that the stable deletion in $P(\bfx)$ has to be $(4, 2, 1)$.
    Correcting the stable deletion yields the permutation $P(\bfx) = (4,5,2,3,1)$ and hence we get back $\bfx$ from $A(\bfx)$ and $P(\bfx)$.
\end{example}

As we saw in Section \ref{sec:bcw_codes}, there is an efficient decoding algorithm for the set codes.
Moreover, there is a decoding algorithm for the construction in \cite{WNCV24} under the assumption that there is a decoding algorithm in the Hamming metric as discussed in Section \ref{sec:stable_del}.

\subsection{Analysis of the Code Size and Redundancy}\label{subsec:analysis}
In the previous section, we constructed codes correcting deletions using set codes and permutation codes.
Since Construction \ref{con1-UD} with unstable deletions can only be used in the case $t=1$, we here analyze the size of the code from Construction \ref{con2-SD} with stable deletions.

From Section \ref{sec:stable_del} we know there exists a permutation code $C_{SD}(n,t)$ of size at least $|C_{SD}(n,t)| \geq \frac{n!}{(2n)^{3t-1}}$ and from Section \ref{sec:bcw_codes} and Section \ref{sec:set_bcw_codes} that there exists a set code $C_{S}(q,n,t)$ of size at least $|C_S(q,n,t)| \geq \frac{\binom{q}{n}} {(2q)^t}$. Hence, via Theorem \ref{thm.mult_free_stable_del} we get a multiplicity-free code correcting $t$ deletions of size at least
\[
    |C_2(q, n, t)| \geq \frac{n!}{(2n)^{3t-1}} \frac{\binom{q}{n}} {(2q)^t} = \left(\prod_{i=0}^{n-1} (q-i) \right) \frac{1}{(2n)^{3t-1}(2q)^t}.
\]
Taking the logarithm with base 2 yields
\begin{align*}
    \left(\sum_{i=0}^{n-1} \log(q-i) \right) - t \log(q) - (3t-1) \log(n) - (4t-1).
\end{align*}
Now, assuming $q > n^{2+\varepsilon}$, we get that for $i \in \{0, \ldots, n-1\}$
\[
    q - i \geq q-n \geq q - \frac{q}{n^{1 + \varepsilon}}
    = q \left(1 - \frac{1}{n^{1 + \varepsilon}} \right).
\]
Thus, $\left(\sum_{i=0}^{n-1} \log(q-i) \right) \geq n \log(q) - O\left(\frac{1}{n^{\varepsilon}} \right).$

The redundancy of the code $C_2(q, n, t)$ is defined as
\[
    \log(|\Sigma_q^n|) - \log(|C_2(q, n, t)|)
    = n \log(q) - \log(|C_2(q, n, t)|) .
\]
Therefore, the redundancy is bounded by
\[
    t \log(q) + (3t - 1) \log(n) + (4t-1) + O\left(\frac{1}{n^{\varepsilon}} \right).
\]
Actually, both the permutation code and the set code constructions use the smallest primes larger than a given number.
If instead of using Bertrand's postulate we use the prime number theorem, the term $4t-1$ can be replaced by $\delta t$ for any $\delta > 0$ for $n$ large enough using the very same argument.

Note that also the term $(3t-1) \log(n)$ could be improved to $2t \log(n)$ using the non-constructive results from \cite{WNCV24} mentioned in Section \ref{sec:stable_del}.

The Singleton bound \cite{liu2022bounds, HaeuplerS21} tells us that
\[
    \log(|C|) \leq (n - t) \log(q).
\]
Our code has size
\begin{align*}
    n \log(q) - \left(t \log(q) + (3t - 1) \log(n) + \delta t + O\left(\frac{1}{n^{\varepsilon}} \right)\right).
\end{align*}
We show that, by increasing the alphabet size, we can make the code size arbitrarily close to the Singleton bound.
Let $t$ be fixed and $q = n^\alpha$ for some $\alpha>2$.
Then our code has size
\[
    \left( n - t - \frac{3t-1}{\alpha} + o(1) \right) \log(q).
\]
If we allow a deviation of $\eta > 0$ from the Singleton bound, i.e., code size $(n-t-\eta) \log(q)$, then choosing $\alpha > \frac{3t-1}{\eta}$ gives the result.
Thus, our code asymptotically comes arbitrarily close to the Singleton bound if we let $\alpha$ be large enough.

\section{Conclusion and Discussion}
In this paper, we provided a new construction of non-binary deletion correcting codes and their decoding for $q>n$.
For large length $n$, alphabet size $q > n^{2 + \varepsilon}$ with $\varepsilon > 0$, and error correction capability $t$, our code has redundancy at most $t \log(q) + (3t - 1) \log(n) + \delta t + O\left(\frac{1}{n^{\varepsilon}} \right)$ for arbitrary $\delta > 0$.

In the literature, there are already several known results on $q$-ary codes correcting $t>1$ deletions and on their redundancy.
For example, codes in \cite{sima2020optimalqary} require $30t \log q$ bits of redundancy, codes in \cite{hagiwara2024pointer} require $5t \log q$ bits of redundancy, and codes in \cite{HaeuplerS21} require at least $t\log q + \Theta(n)$ bits of redundancy.
Hence, when $q>n^{2+\epsilon}$ for $\epsilon>0$ and $t$ is constant, our codes have smaller redundancy.
To the best of our knowledge, in this setting our codes have the biggest size among all known $q$-ary codes of length $n$ correcting $t$ deletions.

Moreover, for alphabet size $q = n^\alpha$ with $\alpha$ large enough, the size of our code is asymptotically arbitrarily close to the Singleton bound.

Important open problems include finding efficient encoding and message recovery algorithms for our code and the decoding of permutation codes in the Hamming metric.

Further possible research directions could be to study the case where also insertion errors can occur and to generalize the construction to sequences that are not necessarily multiplicity-free.

\section*{Acknowledgment}
We would like to thank Joachim Rosenthal and Yeow Meng Chee for the useful inputs, discussions and comments.

Michael Schaller and Beatrice Toesca are supported by the Swiss National Foundation through grant no.\ 212865.

\bibliographystyle{abbrv}
\bibliography{A_New_Construction_of_Non_Binary_Deletion_Correcting_Codes_and_their_Decoding_arxiv}

\end{document}